\documentclass{easychair}
\usepackage{amssymb}
\usepackage{tabularx}
\usepackage{multirow}
\usepackage{tikz}
\usetikzlibrary{positioning}
\usepackage{mathtools}
\usepackage[noend]{algpseudocode}
\usepackage{stmaryrd}
\usepackage{amsmath}

\usepackage{cleveref}
\usepackage{float}
\usepackage{algorithm2e}
\usepackage{colortbl}
\definecolor{gelb}{rgb}{1,0.9,0.5}
\definecolor{blau}{rgb}{0,1,1}

\title{The directed metric dimension of directed co-graphs}
\author{Yannick Schmitz \and Egon Wanke}
\institute{Heinrich-Heine-Universit\"at D\"usseldorf, Germany\\ \email{yannick.schmitz@hhu.de, egon.wanke@hhu.de}}
\newtheorem{theorem}{Theorem}
\newtheorem{definition}{Definition}

\newtheorem{lemma}{Lemma}
\newtheorem{example}{Example}

\newtheorem{notation}{Notation}

\authorrunning{Schmitz, Wanke}
\titlerunning{The directed metric dimension of directed co-graphs}

\newcommand{\DP}[3]{
\begin{center}
\begin{tabularx}{.78\textwidth}{lL}
\hline\hline
\multicolumn{2}{c}{\sc{#1}} \\
\hline
\em{Instance:}& #2\\
\em{Question:}& #3 \\ 
\hline\hline
\end{tabularx}
\end{center}
}

\newcommand{\union}{\hbox{\small$\cup$}\,}
\newcommand{\join}{\hbox{$\times$}\,}
\newcommand{\rjoin}{\raise1pt\hbox{\scriptsize$\gg$}\,}

\newcolumntype{L}{>{\raggedright\arraybackslash}X}
\newcolumntype{R}{>{\raggedleft\arraybackslash}X}
\newcolumntype{C}{>{\centering\arraybackslash}X}

\newcommand{\leftsuc}{\rm left}
\newcommand{\rightsuc}{\rm right}

\newcommand{\vertex}[1]{{\rm vertex}(#1)}
\newcommand{\node}[1]{{\rm node}(#1)}

\newcommand{\capfont}[1]{\textsf{\textit{#1}}}

\begin{document}

\maketitle

\begin{abstract}
A vertex $w$ {\em resolves} two vertices $u$ and $v$ in a directed graph $G$ if the distance from $w$ to $u$ is different to the distance from $w$ to $v$. A set of vertices $R$ is a {\em resolving set} for a directed graph $G$ if for every pair of vertices $u, v$ which are not in $R$ there is at least one vertex in $R$ that resolves $u$ and $v$ in $G$. The {\em directed metric dimension} of a directed graph $G$ is the size of a minimum resolving set for $G$. The decision problem {\sc Directed Metric Dimension} for a given directed graph $G$ and a given number $k$ is the question whether $G$ has a resolving set of size at most $k$. In this paper, we study directed co-graphs. We introduce a linear time algorithm for computing a minimum resolving set for directed co-graphs and show that {\sc Directed Metric Dimension} already is NP-complete for directed acyclic graphs.
\end{abstract}


\section{Introduction}

The metric dimension of graphs is originally defined for undirected graphs. In an undirected graph $G$, two vertices $u$ and $v$ are {\em resolved} by a vertex $w$ if the distance between $w$ and $u$ differs from the distance between $w$ and $v$. In directed graphs, two vertices $u$ and $v$ are {\em resolved} by a vertex $w$ if the distance from $w$ to $u$ differs from the distance from $w$ to $v$. In the undirected case as well as in the directed case, a set of vertices $R$ is called a {\em resolving set} for G if for every pair of vertices $u, v$ which are not in $R$ there is at least one vertex $w$ in $R$ resolving $u$ and $v$. The {\em metric dimension} of $G$ is the size of a smallest resolving set for $G$.

The metric dimension of undirected graphs has been introduced in the 1970s independently by Slater \cite{Sla75} and by Harary and Melter \cite{HM76}. It finds applications in various areas, including network discovery and verification \cite{BEEHHMR05}, geographical routing protocols \cite{LA06}, combinatorial optimisation \cite{ST04}, sensor networks \cite{HW12}, robot navigation \cite{KRR96}, and chemistry \cite{CEJO00,Hay17}. 

Deciding whether a given graph $G$ has metric dimension $\leq k$ is NP-complete for undirected and directed graphs, see \cite{GJ79,KRR96}. There are several algorithms for computing a minimum resolving set in polynomial time for special classes of undirected graphs, as for example for trees \cite{CEJO00,KRR96}, wheels \cite{HMPSCP05}, grid graphs \cite{MT84}, $k$-regular bipartite graphs \cite{BBSSS11}, amalgamation of cycles \cite{IBSS10}, and outerplanar graphs \cite{DPSL12}. The approximability of the metric dimension has been studied for bounded degree, dense, and general graphs in \cite{HSV12}. Upper and lower bounds on the metric dimension are considered in \cite{CGH08,CPZ00} for further classes of undirected graphs. In the undirected case, several variants of the metric dimension have been studied, which are usually NP-complete for general graphs as well \cite{OP07,FR18,SVW21}.

A natural way of generalising graph theoretical problems is to consider their directed counterparts. In the context of the metric dimension of graphs, this was first considered by Chartrand, Rains, and Zhang in \cite{CRZK00}, before receiving further consideration in several other papers, see \cite{FGO06,FXW13,Loz13,RRCM14}.

In this paper, we study directed co-graphs \cite{BDGR97} and introduce a linear time algorithm for computing minimum resolving sets for directed co-graphs in linear time. We also show that {\sc Directed Metric Dimension} already is NP-complete for directed acyclic graphs.


\subsection{The directed metric dimension}

All graphs in this paper are finite and simple. For a graph $G=(V,E)$ with vertex set $V$ and edge set $E$, we write $V(G)$ for vertex set $V$ and $E(G)$ for edge set $E$ to reduce the number of variable names and indices when using several graphs.
For two vertices $u,v \in V(G)$ the distance $d_G(u,v)$ from $u$ to $v$ is the length (number of edges) of a shortest path from $u$ to $v$. If there is no such path from $u$ to $v$, then $d_G(u,v)$ is not defined. A directed graph $G$ is {\em strongly connected} if for each pair of vertices $u$ and $v$, there is a path from $u$ to $v$ and a path from $v$ to $u$.

The {\em metric dimension} of a directed graph can be defined in the same way as for undirected graphs, see \Cref{figure1} for an example.


\begin{definition}[Directed metric dimension]
\label{definition1}
Two distinct vertices $u$ and $v$ of a directed graph $G$ are {\em resolved} by a vertex $w$ if
\begin{enumerate}
\item $w=u$,
\item $w=v$, or
\item there is a path from $w$ to $u$ and a path from $w$ to $v$ such that $d_G(w,u) \not= d_G(w,v)$.
\end{enumerate}

A set of vertices $R \subseteq V(G)$ is called a {\em resolving set} of $G$ if for every pair of vertices $u, v \in V(G)$ there is at least one vertex $w$ in $R$ resolving $u$ and $v$. The {\em directed metric dimension} of $G$ is the size of a minimum resolving set for $G$.
\end{definition}

Note that it is also possible to consider the distances $d_G(u,w)$ from each vertex to the vertices in $R$ instead of the distances $d_G(w,u)$, but both definitions are equivalent if every edge $(u,v)$ in $G$ is replaced by an edge $(v,u)$. Also note that if $d_G(w,u)$ is undefined, it can not be compared to any other distance $d_G(w,v)$.

The decision problem {\sc Directed Metric Dimension} it defined as follows.


\DP
{Directed Metric Dimension}
{A directed graph $G$ and a number $k$.}
{Is there a resolving set $R \subseteq V(G)$ for $G$ of size at most $k$?}


\begin{figure}[htb]
\center
\includegraphics[width=140pt]{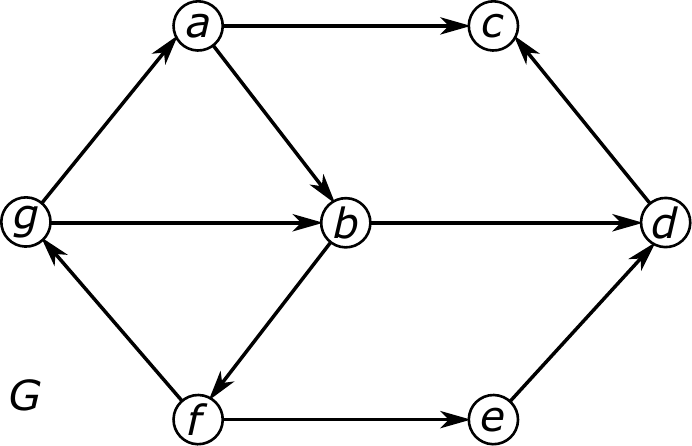}
\caption{Vertex set $\{\capfont{a},\capfont{f},\capfont{g}\}$ is a minimum resolving set for directed graph $G$.
Vertex $\capfont{a}$ resolves the vertex pairs
$\capfont{a},\capfont{b}$,
$\capfont{a},\capfont{c}$,
$\capfont{a},\capfont{d}$,
$\capfont{a},\capfont{e}$,
$\capfont{a},\capfont{f}$,
$\capfont{a},\capfont{g}$,
$\capfont{b},\capfont{d}$,
$\capfont{b},\capfont{e}$,
$\capfont{b},\capfont{f}$,
$\capfont{b},\capfont{g}$,
$\capfont{c},\capfont{d}$,
$\capfont{c},\capfont{e}$,
$\capfont{c},\capfont{f}$,
$\capfont{c},\capfont{g}$,
$\capfont{d},\capfont{e}$,
$\capfont{d},\capfont{g}$,
$\capfont{f},\capfont{e}$, and
$\capfont{f},\capfont{g}$ in $G$,
vertex $\capfont{f}$ resolves the vertex pairs
$\capfont{a},\capfont{c}$,
$\capfont{a},\capfont{e}$,
$\capfont{a},\capfont{f}$,
$\capfont{a},\capfont{g}$,
$\capfont{b},\capfont{c}$,
$\capfont{b},\capfont{e}$,
$\capfont{b},\capfont{f}$,
$\capfont{b},\capfont{g}$,
$\capfont{c},\capfont{d}$,
$\capfont{c},\capfont{e}$,
$\capfont{c},\capfont{f}$,
$\capfont{c},\capfont{g}$,
$\capfont{d},\capfont{e}$,
$\capfont{d},\capfont{f}$,
$\capfont{d},\capfont{g}$,
$\capfont{e},\capfont{f}$, and
$\capfont{f},\capfont{g}$ in $G$, and
vertex $\capfont{g}$ resolves the vertex pairs
$\capfont{a},\capfont{c}$,
$\capfont{a},\capfont{d}$,
$\capfont{a},\capfont{e}$,
$\capfont{a},\capfont{f}$,
$\capfont{a},\capfont{g}$,
$\capfont{b},\capfont{c}$,
$\capfont{b},\capfont{d}$,
$\capfont{b},\capfont{e}$,
$\capfont{b},\capfont{f}$,
$\capfont{b},\capfont{g}$,
$\capfont{c},\capfont{e}$,
$\capfont{c},\capfont{g}$,
$\capfont{d},\capfont{e}$,
$\capfont{d},\capfont{g}$,
$\capfont{e},\capfont{f}$,
$\capfont{e},\capfont{g}$, and
$\capfont{f},\capfont{g}$ in $G$.
}
\label{figure1}
\end{figure}

The following notion of a vertex set $R$ that resolves only the vertex pairs of a vertex set $U \subseteq V(G)$ is frequently used in the next section.

\begin{notation}
Let $U \subseteq V(G)$ be a set of vertices. A set of vertices $R \subseteq U$ {\em resolves a vertex set $U$ in $G$}, or in other words, {\em $R$ is a resolving set for $U$ in $G$}, if for each pair of vertices $u, v \in U$ there is at least one vertex $w$ in $R$ that resolves $u$ and $v$ in $G$.
\end{notation}

If $R$ resolves the vertices of $U$ in $G$, then the necessary shortest paths with different distances $d_G(w,u)$ and $d_G(w,v)$ are paths in graph $G$ and do not need to be paths in the subgraph of $G$ induced by $U$. For example, in \Cref{figure1} vertex set $R = \{\capfont{a}\}$ is a resolving set for $U = \{\capfont{c},\capfont{d},\capfont{g}\}$ in $G$.


\section{Directed metric dimension of directed co-graphs}

In this section, we show how to compute a minimum resolving set for directed co-graphs in linear time.
In \Cref{definition2}, co-graphs and co-trees are defined step by step simultaneously.
We use variable names with a hat symbol for nodes in trees to distinguish them more clearly from vertices in graphs.


\begin{definition}[\bf Directed co-graphs and co-trees]
\label{definition2}
~
\begin{itemize}
\item A directed graph $G$ that consists of a single vertex $u$ is a directed co-graph.

The co-tree $T$ of $G$ consists of a single node $\hat{u}$ associated with vertex $u$ of $G$. Node $\hat{u}$ is the {\em root} of $T$. Let $\vertex{\hat{u}} := u$ and $\node{u} := \hat{u}$. The notation $\vertex{\hat{u}}$ is only defined for leaves $\hat{u}$ of $T$.

\item
If $G_1$ and $G_2$ are two directed co-graphs, then the {\em disjoint union}, {\em join}, or {\em directed join} of $G_1$ and $G_2$, denoted by $G_1\union G_2$, $G_1 \join G_2$, and $G_1 \rjoin G_2$, respectively, is a directed co-graph $G$ with vertex set $V(G_1)\cup V(G_2)$ and edge set
$$
\begin{array}{l}E(G_1)\cup E(G_2), \\
E(G_1)\cup E(G_2)\cup \{(u,v),(v,u)\,\vert\,u\in V(G_1),\,v\in V(G_2)\}, {\text and} \\
E(G_1)\cup E(G_2)\cup \{(u,v)\,\vert\,u\in V(G_1),\,v\in V(G_2)\}, \text{ respectively.}\end{array}
$$
Let $T_1$ and $T_2$ be the co-trees of $G_1$ and $G_2$ with root $\hat{l}$ and $\hat{r}$, respectively. The co-tree $T$ of $G$ is the disjoint union of $T_1$ and $T_2$ with an additional node $\hat{u}$ and two additional edges $\{\hat{u},\hat{l}\}$ and $\{\hat{u},\hat{r}\}$. Node $\hat{u}$ is the {\em root} of $T$ labelled by $\union$, $\join$, and $\rjoin$, respectively.
Node $\hat{l}$ is the {\em left successor node} of $\hat{u}$, also denoted by $\leftsuc(\hat{u})$, and node $\hat{r}$ is the {\em right successor node} of $\hat{u}$, also denoted by $\rightsuc(\hat{u})$. Node $\hat{u}$ is the {\em predecessor node} of $\hat{l}$ and $\hat{r}$.

\end{itemize}

The nodes of $T$ that are not leaves are called {\em inner nodes} of $T$. If $\hat{u}$ is an inner node of $T$ labelled by $\union$ or $\join$ and one of its two successor nodes $\leftsuc(\hat{u})$ and $\rightsuc(\hat{u})$ is a leaf and the other one is not a leaf, then without loss of generality, we assume that successor node $\rightsuc(\hat{u})$ is the leaf.
\end{definition}

\Cref{figure2} shows an example of a directed co-graph and its co-tree. Directed co-graphs can be recognised in linear time \cite{CP06}. This includes the computation of the co-tree. 
In a strongly connected directed co-graph for each pair of vertices $u,v \in V(G)$ the distance $d_G(u,v)$ is at most $2$.


\begin{figure}[htb]
\center
\includegraphics[width=146pt]{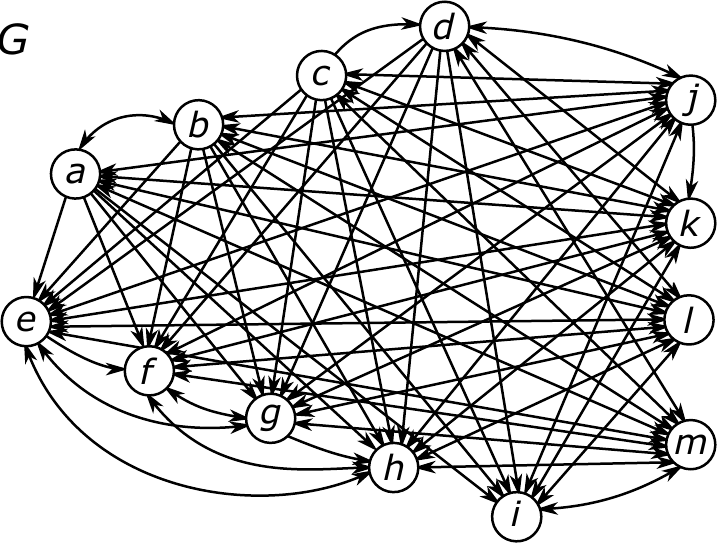}

\medskip
\includegraphics[width=200pt]{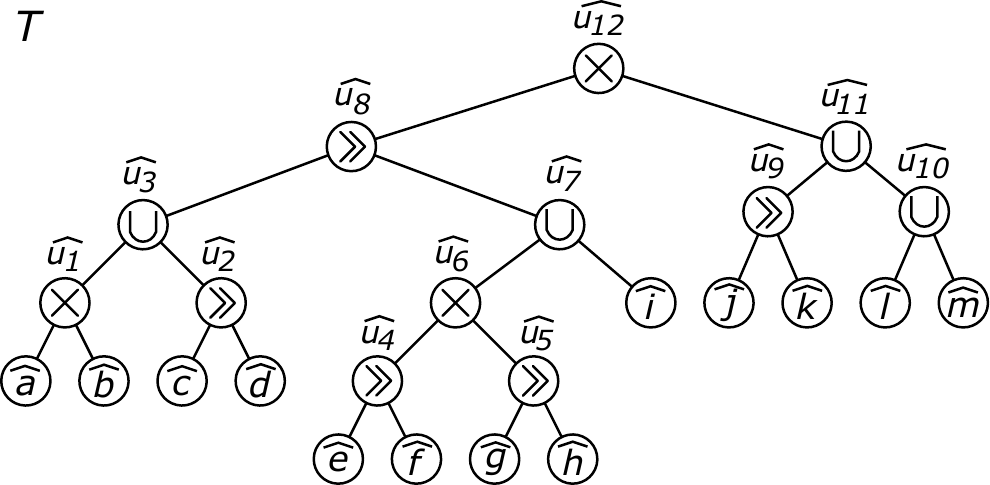}
\caption{A directed co-graph $G$ and its co-tree $T$.}
\label{figure2}
\end{figure}




For a directed co-graph $G$ with co-tree $T$ and a node $\hat{w} \in V(T)$ let $T_{\hat{w}}$ be the subtree of $T$ rooted at $\hat{w}$ and let $G_{\hat{w}}$ be the subgraph of $G$ induced by the vertices $\vertex{\hat{u}}$ of the leaves $\hat{u}$ of $T_{\hat{w}}$.

\begin{lemma}
\label{lemma4}
Let $G$ be a strongly connected directed co-graph with co-tree $T$ and $\hat{w} \in V(T)$ be an inner node of $T$. There is no vertex $w \in V(G) \setminus V(G_{\hat{w}})$ that resolves two vertices of $V(G_{\hat{w}})$ in $G$.
\end{lemma}

\begin{proof}
Since in each further composition of $G_{\hat{w}}$ all vertices of $V(G_{\hat{w}})$ get the same additional connections to vertices of $V(G) \setminus V(G_{\hat{w}})$, it follows that for each vertex $w \in V(G) \setminus V(G_{\hat{w}})$ and each vertex pair $v_1,v_2 \in V(G_{\hat{w}})$ either $d_G(w,v_1) = d_G(w,v_2) = 1$ or  $d_G(w,v_1) = d_G(w,v_2) = 2$. Note that $G$ is strongly connected and thus $d_G(u,v) \leq 2$ for each vertex pair $u,v \in V(G)$.
\end{proof}

If $R$ is a resolving set for $G$, then by \Cref{lemma4} $R \cap V(G_{\hat{w}})$ is a resolving set for $V(G_{\hat{w}})$ in $G$. Note that $R \cap V(G_{\hat{w}})$ does not need to be a resolving set for $G_{\hat{w}}$, because $G_{\hat{w}}$ does not need to be strongly connected. For example, suppose $G_1$, $G_2$, $G_3$, and $G_4$ are four graphs with exactly one vertex $u_1$, $u_2$, $u_3$, and $u_4$, respectively. Then $\{u_1\}$ is a resolving set for $\{u_1,u_2,u_3\}$ in $G=((G_1 \rjoin G_2) \union G_3) \join G_4$, but not a resolving set for $G'=(G_1 \rjoin G_2) \union G_3$, because there is no path from $u_1$ to $u_3$ in $G'$.

Our analysis of directed co-graphs requires the distinction between the following two different vertex types.


\begin{definition}
\label{definition3}
Let $G$ be a directed graph, $u,v \in V(G)$ be two distinct vertices of $G$, and $R \subseteq V(G)$ be a non-empty set of vertices.
A vertex $u \in V(G) \setminus \{R\}$ is called a {\em distance 1 vertex} (or 1-vertex for short) w.r.t.\ $R$ if $\forall w \in R: (w,u) \in E(G)$. It is called a {\em distance 2 vertex} (or 2-vertex for short) w.r.t.\ $R$ if $\forall w \in R: (w,u) \not\in E(G)$.
\end{definition}

\Cref{figure4} shows an example of a 1-vertex and a 2-vertex w.r.t.\ a vertex set $R$.


\begin{figure}[htb]
\center
\includegraphics[width=264pt]{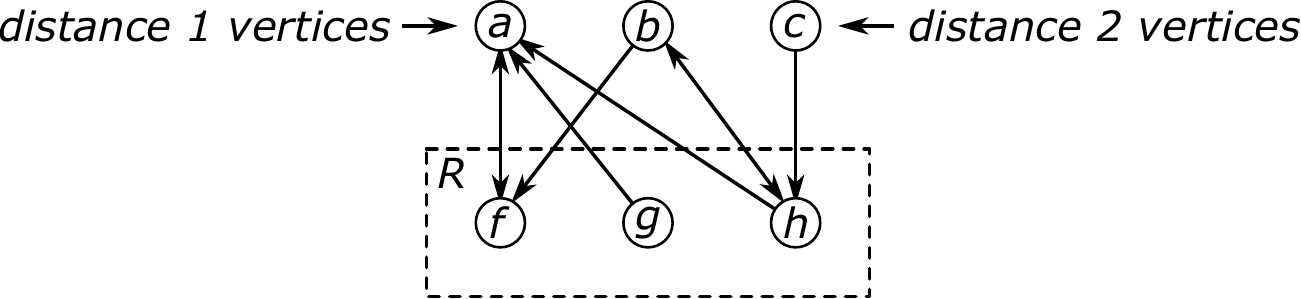}
\caption{Vertex \capfont{a} is a 1-vertices w.r.t.\ $R=\{\capfont{f},\capfont{g},\capfont{h}\}$, vertex \capfont{c} is a 2-vertex w.r.t.\ $R$, and vertex \capfont{b} is neither a 1-vertex nor a 2-vertex w.r.t.\ $R$.}
\label{figure4}
\end{figure}

If a vertex $u \in V(G) \setminus R$ is a 1-vertex or 2-vertex w.r.t.\ vertex set $R \subseteq V(G)$, then $u$ obviously is a 1-vertex or 2-vertex, respectively, w.r.t.\ each non-empty subset $R'$ of $R$. If $R$ is a resolving set for $G$, then there is at most one 1-vertex w.r.t.\ $R$ and at most one 2-vertex w.r.t.\ $R$, because two 1-vertices and two 2-vertices $u,v$  w.r.t.\ $R$ have the same distance $d_G(w,u_1) = d_G(w,u_2) = 1$ and $d_G(w,u_1) = d_G(w,u_2) = 2$, respectively, from each vertex $w \in R$.

The algorithm for computing a minimum resolving set for a strongly connected directed co-graph $G$ analyses the co-tree $T$ for $G$ bottom-up, see procedure {\bf BottomUp}$(\hat{w})$. For each inner node $\hat{w}$ of $T$, it computes one or two minimum resolving sets for $V(G_{\hat{w}})$ in $G$. The resolving sets for node $\hat{w}$ are build by the union of a resolving set $R_{\hat{l}}$ for $V(G_{\hat{l}})$ in $G$ and a resolving set $R_{\hat{r}}$ for $V(G_{\hat{r}})$ in $G$, where $\hat{l}$ and $\hat{r}$ are the left and right successor nodes of $\hat{w}$ in $T$.

\begin{algorithm}[ht]
\SetKwProg{Fn}{BottomUp}{$(\hat{w})$}{}
\Fn{}{
	\If{$(\hat{w} \text{ is an inner node of } T)$}{
		$\hat{l} \leftarrow$ left successor node of $\hat{w}$; \\
		{\bf BottomUp}$(\hat{l})$; \\
		$\hat{r} \leftarrow$ right successor node of $\hat{w}$; \\
		{\bf BottomUp}$(\hat{r})$; \\
		{\bf Merge}$(\hat{w},\hat{l},\hat{r})$;
	}
}
\end{algorithm}

The next lemma specifies the conditions under which the minimum resolving sets for $V(G_{\hat{l}})$ and $V(G_{\hat{r}})$ in $G$ only need to be joined to get a resolving set for $V(G_{\hat{w}})$ in $G$.


\begin{lemma}
\label{lemma6}
Let $G$ be a strongly connected directed co-graph with co-tree $T$, $\hat{w} \in V(T)$ be an inner node of $T$, $\hat{l}= \leftsuc(\hat{w})$ be the left successor of $\hat{w}$ and $\hat{r}=\rightsuc(\hat{w})$ be the right successor of $\hat{w}$.
If $\hat{l}$ is an inner node of $T$, then let $R_{\hat{l}} \subseteq V(G_{\hat{l}})$ be a resolving set for $V(G_{\hat{l}})$ in $G$.
If $\hat{r}$ is an inner node of $T$, then let $R_{\hat{r}} \subseteq V(G_{\hat{r}})$ be a resolving set for $V(G_{\hat{r}})$ in $G$.

We distinguish between the three labels $\union$, $\join$ and $\rjoin$ for node $\hat{w}$ and the cases where the successor nodes of $\hat{w}$ are inner nodes or leaves.

\begin{enumerate}

\item
Let $\hat{w}$ be labelled by $\union$.

\begin{enumerate}
\item
Let $\hat{l}$ and $\hat{r}$ be inner nodes of $T$.

$R_{\hat{l}} \cup R_{\hat{r}}$ is a minimum resolving set for $V(G_{\hat{w}})$ in $G$ iff $R_{\hat{l}}$ is a minimum resolving set for $V(G_{\hat{l}})$ in $G$, $R_{\hat{r}}$ is a minimum resolving set for $V(G_{\hat{r}})$ in $G$, and $G_{\hat{l}}$ has no 2-vertex w.r.t.\ $R_{\hat{l}}$ or $G_{\hat{r}}$ has no 2-vertex w.r.t.\ $R_{\hat{r}}$.

\item
Let $\hat{l}$ be an inner node of $T$ and $\hat{r}$ be a leaf of $T$.

$R_{\hat{l}}$ is a minimum resolving set for $V(G_{\hat{w}})$ in $G$ iff $R_{\hat{l}}$ is a minimum resolving set for $V(G_{\hat{l}})$ in $G$ and $G_{\hat{l}}$ has no 2-vertex w.r.t.\ $R_{\hat{l}}$.

\end{enumerate}

\item
Let $\hat{w}$ be labelled by $\join$.

\begin{enumerate}
\item
Let $\hat{l}$ and $\hat{r}$ be inner nodes of $T$.

$R_{\hat{l}} \cup R_{\hat{r}}$ is a minimum resolving set for $V(G_{\hat{w}})$ in $G$ iff $R_{\hat{l}}$ is a minimum resolving set for $V(G_{\hat{l}})$ in $G$, $R_{\hat{r}}$ is a minimum resolving set for $V(G_{\hat{r}})$ in $G$, and $G_{\hat{l}}$ has no 1-vertex w.r.t.\ $R_{\hat{l}}$ or $G_{\hat{r}}$ has no 1-vertex w.r.t.\ $R_{\hat{r}}$.

\item
Let $\hat{l}$ be an inner node of $T$ and $\hat{r}$ be a leaf of $T$.

$R_{\hat{l}}$ is a minimum resolving set for $V(G_{\hat{w}})$ in $G$ iff
$R_{\hat{l}}$ is a minimum resolving set for $V(G_{\hat{l}})$ in $G$ and
$G_{\hat{l}}$ has no 1-vertex w.r.t.\ $R_{\hat{l}}$.

\end{enumerate}

\item
Let $\hat{w}$ be labelled by $\rjoin$.

\begin{enumerate}
\item
Let $\hat{l}$ and $\hat{r}$ be inner nodes of $T$.

$R_{\hat{l}} \cup R_{\hat{r}}$ is a minimum resolving set for $V(G_{\hat{w}})$ in $G$ iff $R_{\hat{l}}$ is a minimum resolving set for $V(G_{\hat{l}})$ in $G$, $R_{\hat{r}}$ is a minimum resolving set for $V(G_{\hat{r}})$ in $G$, and $G_{\hat{l}}$ has no 1-vertex w.r.t.\ $R_{\hat{l}}$ or $G_{\hat{r}}$ has no 2-vertex w.r.t.\ $R_{\hat{r}}$.

\item
Let $\hat{l}$ be an inner node of $T$ and $\hat{r}$ be a leaf of $T$.

$R_{\hat{l}}$ is a minimum resolving set for $V(G_{\hat{w}})$ in $G$ iff $R_{\hat{l}}$ is a minimum resolving set for $V(G_{\hat{l}})$ in $G$, and $G_{\hat{l}}$ has no 1-vertex w.r.t.\ $R_{\hat{l}}$. 

\item
Let $\hat{l}$ be a leaf of $T$ and $\hat{r}$ be an inner node of $T$.

$R_{\hat{r}}$ is a minimum resolving set for $V(G_{\hat{w}})$ in $G$ iff $R_{\hat{r}}$ is a minimum resolving set for $V(G_{\hat{r}})$ in $G$, and $G_{\hat{r}}$ has no 2-vertex w.r.t.\ $R_{\hat{r}}$. 

\end{enumerate}

\end{enumerate}
\end{lemma}

\begin{proof}
Since two distinct vertices from $V(G_{\hat{l}}) \setminus R_{\hat{l}}$ as well as two distinct vertices from $V(G_{\hat{r}}) \setminus R_{\hat{r}}$ are resolved in $G$ by a vertex of $R_{\hat{r}}$ and a vertex of $R_{\hat{l}}$, respectively, we only need to consider the case in which one vertex is from $V(G_{\hat{l}}) \setminus R_{\hat{l}}$ and the other vertex is from $V(G_{\hat{r}}) \setminus R_{\hat{r}}$.

\begin{enumerate}
\item
Let $\hat{w}$ be labelled by $\union$.

In this case, for each vertex $u_l \in V(G_{\hat{l}}) \setminus R_{\hat{l}}$ and for each vertex $u_r \in V(G_{\hat{r}}) \setminus R_{\hat{r}}$, or $u_r=\vertex{\hat{r}}$ if $\hat{r}$ is a leaf in $T$, we have $d_G(u_l,u_r) = d_G(u_r,u_l) =2$.

\begin{enumerate}
\item
Let $\hat{l}$ and $\hat{r}$ be inner nodes of $T$.

\smallskip
If $u_l$ is not a 2-vertex w.r.t.\ $R_{\hat{l}}$, then there is vertex $v_l \in R_{\hat{l}}$ such that $d_G(v_l,u_l) = 1$.
If $u_r$ is not a 2-vertex w.r.t.\ $R_{\hat{r}}$, then there is a vertex $v_r \in R_{\hat{r}}$ such that $d_G(v_r,u_r) = 1$.
Since $d_G(v_l,u_r) = 2$ and $d_G(v_r,u_l) = 2$, vertex $u_l$ and vertex $u_r$ are resolved in $G$ by $v_l$ or $v_r$.
Thus $R_{\hat{l}} \cup R_{\hat{r}}$ is a resolving set for $V(G_{\hat{w}})$ in $G$.
By \Cref{lemma4}, it is a minimum resolving set for $V(G_{\hat{w}})$ in $G$ if $R_{\hat{l}}$ is a minimum resolving set for $V(G_{\hat{l}})$ in $G$ and $R_{\hat{r}}$ is a minimum resolving set for $V(G_{\hat{l}})$ in $G$.

\smallskip
If $u_l$ is a 2-vertex w.r.t.\ $R_{\hat{l}}$ and $u_r$ is a 2-vertex w.r.t.\ $R_{\hat{r}}$, then for each vertex $v \in R_{\hat{l}} \cup R_{\hat{r}}$, $d_G(v,u_l) = d_G(v,u_r) = 2$.
In this case, $u_l$ and $u_r$ are not resolved by any vertex of $R_{\hat{l}} \cup R_{\hat{r}}$ in $G$ and thus $R_{\hat{l}} \cup R_{\hat{r}}$ is not a resolving set for $V(G_{\hat{w}})$ in $G$.

\item
Let $\hat{l}$ be an inner node of $T$ and $\hat{r}$ be a leaf of $T$.

\smallskip
If $u_l$ is not a 2-vertex w.r.t.\ $R_{\hat{l}}$, then there is a vertex $v_l \in R_{\hat{l}}$ such that $d_G(v_l,u_l) = 1$. Since $d_G(v_l,u_r) = 2$ , vertex $u_l$ and vertex $u_r$ are resolved in $G$ by $v_l$.
Thus $R_{\hat{l}}$ is a resolving set for $V(G_{\hat{w}})$ in $G$.
By \Cref{lemma4}, it is a minimum resolving set for $V(G_{\hat{w}})$ in $G$ if $R_{\hat{l}}$ is a minimum resolving set for $V(G_{\hat{l}})$ in $G$.

\smallskip
If $u_l$ is a 2-vertex w.r.t.\ $R_{\hat{l}}$, then for each vertex $v_l \in R_{\hat{l}}$, $d_G(v_l,u_l) = 2$.
In this case, vertex $u_l$ and vertex $u_r$ are not resolved by any vertex of $R_{\hat{l}}$ and thus $R_{\hat{l}}$ is not a resolving set for $V(G_{\hat{w}})$ in $G$.

\end{enumerate}

\item
Let $\hat{w}$ be labelled by $\join$.

The proof is running analogously to the proof for label $\union$.

\item
Let $\hat{w}$ be labelled by $\rjoin$.

In this case, for each vertex $u_l \in V(G_{\hat{l}})$, or $u_l=\vertex{\hat{l}}$ if $\hat{l}$ is a leaf in $T$, and for each vertex $u_r \in V(G_{\hat{r}})$, or $u_r=\vertex{\hat{r}}$ if $\hat{r}$ is a leaf in $T$, we have $d_G(u_l,u_r) = 1$ and $d_G(u_r,u_l) =2$.

\begin{enumerate}
\item
Let $\hat{l}$ and $\hat{r}$ be inner nodes of $T$.

\smallskip
If $u_l$ is not a 1-vertex w.r.t.\ $R_{\hat{l}}$, then there is a vertex $v_l \in R_{\hat{l}}$ such that $d_G(v_l,u_l)=2$.
If $u_r$ is not a 2-vertex w.r.t.\ $R_{\hat{r}}$, then there is a vertex $v_r \in R_{\hat{r}}$ such that $d_G(v_r,u_r)=1$.
Since $d_G(v_l,u_r) = 1$ and $d_G(v_r,u_l) = 2$, vertex $u_l$ and vertex $u_r$ are resolved in $G$ by $v_l$ or $v_r$.
Thus $R_{\hat{l}} \cup R_{\hat{r}}$ is a resolving set for $V(G_{\hat{w}})$ in $G$.
By \Cref{lemma4}, it is a minimum resolving set for $V(G_{\hat{w}})$ in $G$ if $R_{\hat{l}}$ is a minimum resolving set for $V(G_{\hat{l}})$ in $G$ and $R_{\hat{r}}$ is a minimum resolving set for $V(G_{\hat{l}})$ in $G$.

\smallskip
If $u_l$ is a 1-vertex w.r.t.\ $R_{\hat{l}}$ and $u_r$ is a 2-vertex w.r.t.\ $R_{\hat{r}}$, then for every $v \in R_{\hat{l}}$, $d_G(v,u_l) = d_G(v,u_r) = 1$ and for every $v \in R_{\hat{r}}$, $d_G(v,u_l) = d_G(v,u_r) = 2$. In this case $u_l$ and $u_r$ are not resolved by any vertex of $R_{\hat{l}} \cup R_{\hat{r}}$ and thus $R_{\hat{l}} \cup R_{\hat{r}}$ is not a resolving set for $V(G_{\hat{w}})$ in $G$.

\item
Let $\hat{l}$ be an inner node of $T$ and $\hat{r}$ be a leaf of $T$.

\smallskip
If $u_l$ is not a 1-vertex w.r.t.\ $R_{\hat{l}}$, then there is a vertex $v_l \in R_{\hat{l}}$ such that $d_G(v_l,u_l) = 2$. Since $d_G(v_l,u_r) = 1$, vertex $u_l$ and vertex $u_r$ are resolved in $G$ by $v_l$.
Thus $R_{\hat{l}}$ is a resolving set for $V(G_{\hat{w}})$ in $G$.
By \Cref{lemma4}, it is a minimum resolving set for $V(G_{\hat{w}})$ in $G$ if $R_{\hat{l}}$ is a minimum resolving set for $V(G_{\hat{l}})$ in $G$.

\smallskip
If $u_l$ is a 1-vertex w.r.t.\ $R_{\hat{l}}$, then for each vertex $v_l \in R_{\hat{l}}$, $d_G(v_l,u_l) = 1$.
In this case, vertex $u_l$ and vertex $u_r$ are not resolved by any vertex of $R_{\hat{l}}$ and thus $R_{\hat{l}}$ is not a resolving set for $V(G_{\hat{w}})$ in $G$.

\item
Let $\hat{l}$ be a leaf of $T$ and $\hat{r}$ be an inner node of $T$.

\smallskip
If $u_r$ is not a 2-vertex w.r.t.\ $R_{\hat{r}}$, then there is a vertex $v_r \in R_{\hat{r}}$ such that $d_G(v_r,u_r) = 1$. Since $d_G(v_r,u_l) = 2$, vertex $u_l$ and vertex $u_r$ are resolved in $G$ by $v_r$.
Thus $R_{\hat{r}}$ is a resolving set for $V(G_{\hat{w}})$ in $G$.
By \Cref{lemma4}, it is a minimum resolving set for $V(G_{\hat{w}})$ in $G$ if $R_{\hat{r}}$ is a minimum resolving set for $V(G_{\hat{r}})$ in $G$.

\smallskip
If $u_r$ is a 2-vertex w.r.t.\ $R_{\hat{r}}$, then for each vertex $v_r \in R_{\hat{r}}$, $d_G(v_r,u_r) = 2$.
In this case, vertex $u_r$ and vertex $u_l$ are not resolved by any vertex of $R_{\hat{r}}$ and thus $R_{\hat{r}}$ is not a resolving set for $V(G_{\hat{w}})$ in $G$.

\end{enumerate}

\end{enumerate}

\end{proof}

The proof of \Cref{lemma6} shows the following observation.
\begin{enumerate}
\item
Suppose $\hat{l}$ and $\hat{r}$ are inner node of $T$. If $R_{\hat{l}} \cup R_{\hat{r}}$ is not a resolving set for $V(G_{\hat{w}})$ in $G$, then it is sufficient to add one additional vertex $u$ of $G_{\hat{w}}$ to $R_{\hat{l}} \cup R_{\hat{r}}$ to get a resolving set for $V(G_{\hat{w}})$ in $G$. This is possible with a vertex of $G_{\hat{l}}$ and with a vertex from $G_{\hat{r}}$. If $\hat{w}$ is labelled by $\union$ ($\join$, $\rjoin$), then we can use the 2-vertex (1-vertex, 1-vertex) of $G_{\hat{l}}$ and the (2-vertex, 1-vertex, 2-vertex) of $G_{\hat{r}}$.
\item
Suppose $\hat{l}$ is an inner node and $\hat{r}$ is a leaf of $T$. If $R_{\hat{l}}$ is not a resolving set for $V(G_{\hat{w}})$ in $G$, then it is also sufficient to add one additional vertex $u$ of $G_{\hat{w}}$ to $R_{\hat{l}}$ to get a resolving set for $V(G_{\hat{w}})$ in $G$. If $\hat{w}$ is labelled by $\union$ ($\join$, $\rjoin$), then we can use the 2-vertex (1-vertex, 1-vertex) of $G_{\hat{l}}$ or vertex $\vertex{\hat{r}}$.
\item
Suppose $\hat{l}$ is a leaf and $\hat{r}$ is an inner node of $T$. If $R_{\hat{r}}$ is not a resolving set for $V(G_{\hat{w}})$ in $G$, then again it is also sufficient to add one additional vertex $u$ of $G_{\hat{w}}$ to $R_{\hat{r}}$ to get a resolving set for $V(G_{\hat{w}})$ in $G$. Since $\hat{w}$ is labelled by $\rjoin$, we can use vertex $\vertex{\hat{l}}$ and the 2-vertex of $G_{\hat{r}}$.
\end{enumerate}

If a strongly connected graph $G$ has a 1-vertex $u_1$ (a 2-vertex $u_2$) w.r.t.\ a resolving set $R$, then, obviously, $G$ has no 1-vertex w.r.t.\ $R \cup \{u_1\}$ (no 2-vertex w.r.t.\ $R \cup \{u_2\}$, respectively). Suppose $G$ has both a 1-vertex $u_1$ and a 2-vertex $u_2$ w.r.t.\ a resolving set $R$. If there is an edge from $u_1$ to $u_2$ (no edge from $u_2$ to $u_1$), then $G$ neither has a 1-vertex nor a 2-vertex w.r.t.\ $R \cup \{u_1\}$ (w.r.t.\ $R \cup \{u_2\}$, respectively). In this case, we call vertex $u_1$ (vertex $u_2$, respectively) a {\em double remover}.

\begin{notation}
Let $u_1$ be a 1-vertex and $u_2$ be a 2-vertex of a strongly connected graph $G$ w.r.t.\ a vertex set $R$. A vertex $v$ is a {\em double remover} if $u_1$ is not a 1-vertex and $u_2$ is not a 2-vertex of $G$ w.r.t.\ $R \cup \{v\}$.
\label{notation2}
\end{notation}

\Cref{figure5} (1.), (2.), and (3.) show the necessary and forbidden edges for the case that a 1-vertex $u_1$ is a double remover, a 2-vertex $u_2$ is a double remover, and a 1-vertex $u_1$ and a 2-vertex $u_2$ are both double removers. An induced subgraph with all necessary edges and without all of the forbidden edges of the graph in the third case is not a directed co-graph.


\begin{figure}[htb]
\center
\includegraphics[width=366pt]{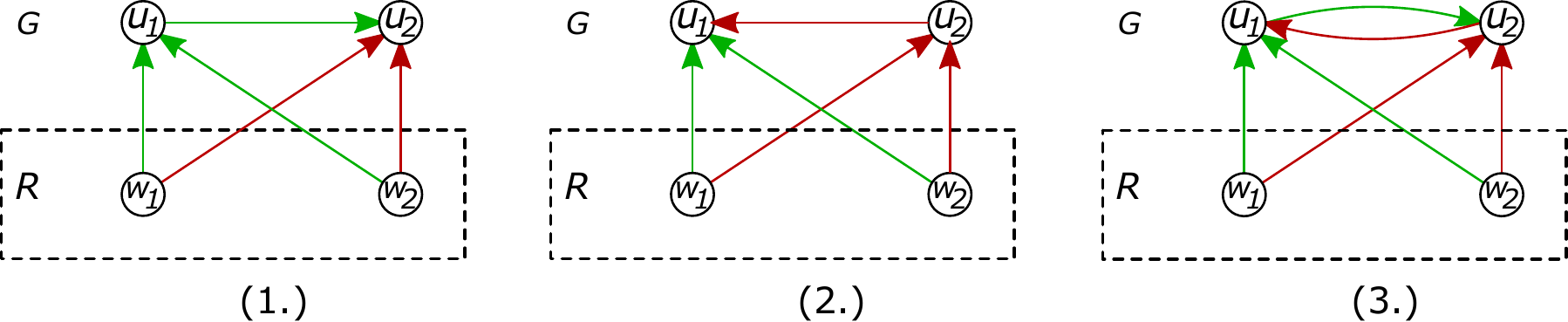} 
\caption{Let $R$ be a resolving set for $V(G)$ in $G$, vertex $u_1$ be a 1-vertex w.r.t.\ $R$,  and vertex $u_2$ be a 2-vertex w.r.t.\ $R$.
In graph (1.), vertex $u_2$ is not a 2-vertex w.r.t.\  $R \cup \{u_1\}$.
In graph (2.), vertex $u_1$ is not a 1-vertex w.r.t.\  $R \cup \{u_2\}$.
In graph (3.), vertex $u_1$ is not a 1-vertex w.r.t.\  $R \cup \{u_2\}$ and vertex $u_2$ is not a 2-vertex w.r.t.\  $R \cup \{u_1\}$.
If these conditions are met, the green edges must be edges of $E(G)$ and the red edges must not be edges of $E(G)$.
There is no directed co-graph that has an induced subgraph with vertex set $\{u_1,u_2,w_1,w_2\}$ and all the necessary green edges but none of the forbidden red edges of graph (3.).}
\label{figure5}
\end{figure}

The next lemma shows that there is no vertex $v \in V(G_{\hat{w}}) \setminus \{u_1,u_2\}$ such that $G_{\hat{w}}$ neither has a 1-vertex nor a 2-vertex w.r.t.\ $R_{\hat{w}} \cup \{v\}$. That is, $u_1$ and $u_2$ are the only possible double removers. This shows that we only need to consider a 1-vertex or a 2-vertex to remove a 1-vertex or 2-vertex of $G_{\hat{w}}$ w.r.t.\ a resolving set $R_{\hat{w}}$, see \Cref{figure6}.


\begin{figure}[htb]
\center
\includegraphics[width=110pt]{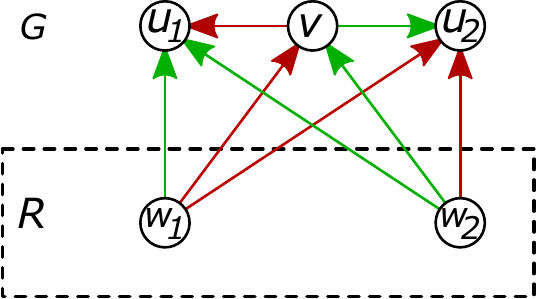} 
\caption{Let $R$ be a resolving set for $V(G)$ in $G$, vertex $u_1$ be a 1-vertex w.r.t.\ $R$, vertex $u_2$ be a 2-vertex w.r.t.\ $R$, and vertex $v$ be a vertex such that $u_1$ is not a 1-vertex w.r.t.\ $R \cup \{v\}$ and $u_2$ is not a 2-vertex w.r.t.\ $R \cup \{v\}$.
Then the green edges must be edges of $E(G)$ and the red edges must not be edges of $E(G)$.
There is no directed co-graph that has an induced subgraph with vertex set $\{u_1,v,u_2,w_1,w_2\}$ and all the necessary green edges but none of the forbidden red edges.}
\label{figure6}
\end{figure}


\begin{lemma}
\label{lemma5}
Let $G$ be a strongly connected directed co-graph with co-tree $T$, $\hat{w} \in V(T)$ be an inner node of $T$, and $R_{\hat{w}} \subseteq V(G_{\hat{w}})$ be a resolving sets for $V(G_{\hat{w}})$ in $G$, such that $G_{\hat{w}}$ has a 1-vertex $u_1$ and a 2-vertex $u_2$ w.r.t.\ $R_{\hat{w}}$. Then there is no vertex $v \in V(G_{\hat{w}}) \setminus \{u_1,u_2\}$ such that $R_{\hat{w}} \cup \{v\}$ is a resolving set for $V(G_{\hat{w}})$ in $G$ and $G_{\hat{w}}$ has neither a 1-vertex nor a 2-vertex w.r.t.\ $R_{\hat{w}} \cup \{v\}$.
\end{lemma}

\begin{proof}
[Proof by contradiction]
Suppose there is a vertex $v \in V(G_{\hat{w}}) \setminus \{u_1,u_2\}$ such that $R_{\hat{w}} \cup \{v\}$ is a resolving set for $V(G_{\hat{w}})$ in $G$ and $G_{\hat{w}}$ neither has a 1-vertex nor a 2-vertex w.r.t.\ $R_{\hat{w}} \cup \{v\}$. Since $R_{\hat{w}}$ is a resolving set for $V(G_{\hat{w}})$ in $G$ the three vertex pairs $u_1,u_2$, $u_1,v$ and $u_2,v$ have to be resolved in $G$ by vertices of $R_{\hat{w}}$. Since each vertex of $R_{\hat{w}}$ resolves the vertex pair $u_1,u_2$ in $G$, set $R_{\hat{w}}$ has to contain at least two vertices $w_1,w_2$, where $w_1$ has no edge to $v$ and $w_2$ has an edge to $v$. That is, vertex $w_1$ resolves vertex pair $u_1,v$ in $G$ and vertex $w_2$ resolves vertex pair $u_2,v$ in $G$. Up to now, we have $(w_1,u_1),(w_2,u_1),(w_2,v) \in E(G_{\hat{w}})$ and $(w_1,u_2),(w_1,v),(w_2,u_2) \not\in E(G_{\hat{w}})$. Since $u_1$ is not a 1-vertex w.r.t.\ $R_{\hat{w}} \cup \{v\}$, we have $(v,u_1) \not\in E(G_{\hat{w}})$. Since $u_2$ is not a 2-vertex w.r.t.\ $R_{\hat{w}} \cup \{v\}$, we have $(v,u_2) \in E(G_{\hat{w}})$. \Cref{figure6} illustrates this situation.

Now it is easy to see that there is no partition of the vertex set $\{u_1,u_2,v,w_1,w_2\}$ into two non-empty sets $U_1$ and $U_2$ such that there are graphs $G_1$ and $G_2$ with vertex sets $U_1$ and $U_2$, respectively, and $G_1 \cup G_2$, $G_1 \join G_2$, or $G_1 \rjoin G_2$ contains the required edges and does not contain the forbidden edges.
\end{proof}

For each inner node $\hat{w}$ of $T$, the procedure {\bf Merge}$(\hat{w},\hat{l},\hat{r})$ computes at most two of the following 4 minimum resolving sets for $V(G_{\hat{w}})$ in $G$. The 4 sets are denoted by $R_{\hat{w},t_{\hat{w}}}$, where $t_{\hat{w}} \in \{0,1,2,12\}$ and
\begin{enumerate}
\item
$R_{\hat{w},0}$ is a minimum resolving set for $V(G_{\hat{w}})$ in $G$ such that $G_{\hat{w}}$ neither has a 1-vertex nor a 2-vertex w.r.t.\ $R_{\hat{w},0}$,
\item
$R_{\hat{w},1}$ is a minimum resolving set for $V(G_{\hat{w}})$ in $G$ such that $G_{\hat{w}}$ has a 1-vertex but no 2-vertex w.r.t.\ $R_{\hat{w},1}$,
\item
$R_{\hat{w},2}$ is a minimum resolving set for $V(G_{\hat{w}})$ in $G$ such that $G_{\hat{w}}$ has no 1-vertex but a 2-vertex w.r.t.\ $R_{\hat{w},2}$, and
\item
$R_{\hat{w},12}$ is a minimum resolving set for $V(G_{\hat{w}})$ in $G$ such that $G_{\hat{w}}$ has a 1-vertex and a 2-vertex w.r.t.\ $R_{\hat{w},12}$.
\end{enumerate}

Additionally, we store for each defined set $R_{\hat{w},1}$ and $R_{\hat{w},2}$, the 1-vertex and 2-vertex of $G_{\hat{w}}$ w.r.t.\ $R_{\hat{w},1}$ and $R_{\hat{w},2}$, respectively, as well as for each defined set $R_{\hat{w},12}$, the 1-vertex and 2-vertex of $G_{\hat{w}}$ w.r.t.\ $R_{\hat{w},12}$ and the information whether one of these two vertices, and if so, which one of them, is a double remover.

The minimum resolving sets for $V(G_{\hat{w}})$ in $G$ are computed from the previously computed minimum resolving sets for $V(G_{\hat{l}})$ and $V(G_{\hat{r}})$ in $G$. In the simplest case, a minimum resolving set for $V(G_{\hat{w}})$ in $G$ can be defined by the union of a minimum resolving set for $V(G_{\hat{l}})$ in $G$ and a minimum resolving set for $V(G_{\hat{r}})$ in $G$. In the worst case, at most one additional vertex from $V(G_{\hat{w}})$ has to be added. However, this additional vertex $v$ is always a 1-vertex or 2-vertex of $G_{\hat{l}}$ w.r.t.\ the corresponding minimum resolving set for $V(G_{\hat{l}})$ in $G$, or a 1-vertex or 2-vertex of $G_{\hat{r}}$ w.r.t.\ the corresponding minimum resolving set for $V(G_{\hat{r}})$ in $G$. Existing double removers are of course preferred here. If $\hat{l}$ or $\hat{r}$ or both are leaves, then vertex $\vertex{\hat{l}}$ and $\vertex{\hat{r}}$ are also a possible choice for $v$.

To explain how the procedure {\bf Merge}$(\hat{w},\hat{l},\hat{r})$ works in detail, we again distinguish between the three labels $\union$, $\join$ and $\rjoin$ for node $\hat{w}$ and the cases where the successor nodes of $\hat{w}$ are inner nodes or leaves. Let $R_{\hat{l},t_{\hat{l}}}$ be an already computed resolving set for $V(G_{\hat{l}})$ in $G$ and $R_{\hat{r},t_{\hat{r}}}$ be an already computed resolving set for $V(G_{\hat{r}})$ in $G$, where $t_{\hat{l}},t_{\hat{r}} \in \{0,1,2,12\}$. To update a resolving set $R_{\hat{w},t_{\hat{w}}}$ for some $t_{\hat{w}} \in \{0, 1, 2, 12\}$ by a minimum resolving set $R$ means that $R_{\hat{w},t_{\hat{w}}}$ is set to $R$, if $R_{\hat{w},t_{\hat{w}}}$ is not already defined, or the size of $R$ is less than the size of $R_{\hat{w},t_{\hat{w}}}$, if $R_{\hat{w},t_{\hat{w}}}$ is already defined.

Suppose $R_{\hat{w},0}$, $R_{\hat{w},1}$, $R_{\hat{w},2}$, and $R_{\hat{w},12}$ are defined. The following simplifications for updating the quantities result from the fact that we can prefer smaller quantities or quantities of the same size with fewer restrictions.
\begin{itemize}
\item If $|R_{\hat{w},0}| \leq |R_{\hat{w},1}|$ or $|R_{\hat{w},0}| \leq |R_{\hat{w},2}|$, then we do not need to update $R_{\hat{w},1}$ or $R_{\hat{w},2}$.
\item If $|R_{\hat{w},1}| < |R_{\hat{w},0}|$ or $|R_{\hat{w},2}| < |R_{\hat{w},0}|$, then we do not need to update $R_{\hat{w},0}$.
\item If $|R_{\hat{w},1}| \leq |R_{\hat{w},12}|$ or $|R_{\hat{w},2}| \leq |R_{\hat{w},12}|$, then we do not need to update $R_{\hat{w},12}$.
\item If $|R_{\hat{w},1}| < |R_{\hat{w},12}|$ or $|R_{\hat{w},2}| < |R_{\hat{w},12}|$, then we do not need to update $R_{\hat{w},1}$ or $R_{\hat{w},2}$.
\end{itemize}
That is, we either consider the resolving set $R_{\hat{w},0}$, one or both of the resolving sets $R_{\hat{w},1}$ and $R_{\hat{w},2}$, or the resolving set $R_{\hat{w},12}$ for $V(G_{\hat{w}})$ in $G$.

\begin{table}[hbt]
\setlength{\arraycolsep}{4pt}
\null \hfill
$
\begin{array}[t]{|cc|c|ccc|}
\hline
t_{\hat{l}} & t_{\hat{r}} & v & t_{\hat{w}}   & u_1 & u_2 \\
\hline
0   & 0   & - & 0   & - & - \\
\hline
0   & 1   & - & 0   & - & - \\
\hline
0   & 2   & - & 2   & - & u_{r,2} \\
\hline
0   & 12  & - & 2   & - & u_{r,2} \\
\hline
1   & 0   & - & 0   & - & - \\
\hline
1   & 1   & - & 0   & - & - \\
\hline
1   & 2   & - & 2   & - & u_{r,2} \\
\hline
1   & 12  & - & 2   & - & u_{r,2} \\
\hline
2   & 0   & - & 2   & - & u_{l,2} \\
\hline
2   & 1   & - & 2   & - & u_{l,2} \\
\hline
2   & 2   & u_{l,2} & 2 & - & u_{r,2} \\
\cline{3-6}
    &     & \cellcolor{blau} u_{r,2} & \cellcolor{blau} 2 & \cellcolor{blau} - & \cellcolor{blau} u_{l,2} \\
\hline
2   & 12  & u_{l,2} & 2 & - & u_{r,2} \\
\cline{3-6}
 & & \cellcolor{blau} u_{r,2} & \cellcolor{blau} 2 & \cellcolor{blau} - & \cellcolor{blau} u_{l,2} \\
\cline{3-6}
 & & \cellcolor{blau} u_{r,\star} & \cellcolor{blau} 2 & \cellcolor{blau} - & \cellcolor{blau} u_{l,2} \\
\hline
12  & 0   & - & 2   & - & u_{l,2} \\
\hline
12  & 1   & - & 2   & - & u_{l,2} \\
\hline
12  & 2   & u_{l,2} & 2 & - & u_{r,2} \\
\cline{3-6}
 & & \cellcolor{blau} u_{l,\star} & \cellcolor{blau} 2 & \cellcolor{blau} - & \cellcolor{blau} u_{r,2} \\
\cline{3-6}
 & & \cellcolor{blau} u_{r,2} & \cellcolor{blau} 2 & \cellcolor{blau} - & \cellcolor{blau} u_{l,2} \\
\hline
12 & 12 & u_{l,2} & 2 & - & u_{r,2} \\
\cline{3-6}
 & & \cellcolor{blau} u_{l,\star} & \cellcolor{blau} 2 & \cellcolor{blau} - & \cellcolor{blau} u_{r,2} \\
\cline{3-6}
 & & \cellcolor{blau} u_{r,2} & \cellcolor{blau} 2 & \cellcolor{blau} - & \cellcolor{blau} u_{l,2} \\
\cline{3-6}
 & & \cellcolor{blau} u_{r,\star} & \cellcolor{blau} 2 & \cellcolor{blau} - & \cellcolor{blau} u_{l,2} \\
\hline
\multicolumn{4}{c}{ } \\[-6pt]
\multicolumn{4}{c}{\text{Label }\union} \\
\end{array}
$
\hfill
$
\begin{array}[t]{|cc|c|ccc|}
\hline
t_{\hat{l}} & t_{\hat{r}} & v & t_{\hat{w}}   & u_1 & u_2 \\
\hline
0   & 0   & - & 0 & - & - \\
\hline
0   & 1   & - & 1 & u_{r,1} & - \\
\hline
0   & 2   & - & 0 & - & - \\
\hline
0   & 12  & - & 1 & u_{r,1} & - \\
\hline
1   & 0   & - & 1 & u_{l,1} & - \\
\hline
1   & 1   & u_{l,1} & 1 & u_{r,1} & - \\
\cline{3-6}
 & & \cellcolor{blau} u_{r,1} & \cellcolor{blau} 1 & \cellcolor{blau} u_{l,1} & \cellcolor{blau} - \\
\hline
1   & 2   & - & 1 & u_{l,1} & - \\
\hline
1 & 12 & u_{l,1} & 1 & u_{r,1} & - \\
\cline{3-6}
 & & \cellcolor{blau} u_{r,1} & \cellcolor{blau} 1 & \cellcolor{blau} u_{l,1} & \cellcolor{blau} - \\
\cline{3-6}
 & & \cellcolor{blau} u_{r,\star} & \cellcolor{blau} 1 & \cellcolor{blau} u_{l,1} & \cellcolor{blau} - \\
\hline
2   & 0   & - & 0   & - & - \\
\hline
2   & 1   & - & 1   & u_{r,1} & - \\
\hline
2   & 2   & - & 0   & - & - \\
\hline
2   & 12 & - & 1   & u_{r,1} & - \\
\hline
12 & 0 & - & 1 & u_{l,1} & - \\
\hline
12 & 1 & u_{l,1} & 1 & u_{r,1} & - \\
\cline{3-6}
 & & \cellcolor{blau} u_{l,\star} & \cellcolor{blau} 1 & \cellcolor{blau} u_{r,1} & \cellcolor{blau} - \\
\cline{3-6}
 & & \cellcolor{blau} u_{r,1} & \cellcolor{blau} 1 & \cellcolor{blau} u_{l,1} & \cellcolor{blau} - \\
\hline
12 & 2 & - & 1 & u_{l,1} & - \\
\hline
12 & 12 & u_{l,1} & 1 & u_{r,1} & - \\
\cline{3-6}
 & & \cellcolor{blau} u_{l,\star} & \cellcolor{blau} 1 & \cellcolor{blau} u_{r,1} & \cellcolor{blau} - \\
\cline{3-6}
 & & \cellcolor{blau} u_{r,1} & \cellcolor{blau} 1 & \cellcolor{blau} u_{l,1} &\cellcolor{blau}  -  \\
\cline{3-6}
 & & \cellcolor{blau} u_{r,\star} & \cellcolor{blau} 1 & \cellcolor{blau} u_{l,1} & \cellcolor{blau} -  \\
\hline
\multicolumn{6}{c}{ } \\[-6pt]
\multicolumn{6}{c}{\text{Label }\join} 
\end{array}
$
\hfill
$
\begin{array}[t]{|cc|c|ccc|}
\hline
t_{\hat{l}} & t_{\hat{r}} & v & t_{\hat{w}} & u_1 & u_2 \\
\hline
0 & 0 & - & 0 & - & - \\
\hline
0 & 1 & - & 1 & u_{r,1} & - \\
\hline
0 & 2 & - & 0 & - & - \\
\hline
0 & 12 & - & 1 & u_{r,1} & - \\
\hline
1 & 0 & - & 0 & - & - \\
\hline
1 & 1 & - & 1 & u_{r,1} & - \\
\hline
1 & 2 & u_{l,1} & 0 & - & - \\
\cline{3-6}
 & & \cellcolor{blau} u_{r,2} & \cellcolor{blau} 0 & \cellcolor{blau} - & \cellcolor{blau} - \\
\hline
1   & 12 & u_{l,1} & 1 & u_{r,1} & -  \\
\cline{3-6}
 & & \cellcolor{blau} u_{r,2} & \cellcolor{blau} 1 & \cellcolor{blau} u_{r,1} & \cellcolor{blau} - \\
\cline{3-6}
 & & u_{r,\star} & 0 & - & - \\
\hline
2 & 0 & - & 2 & - & u_{l,2} \\
\hline
2 & 1 & - & 12 & u_{r,1} & u_{l,2} \\
\hline
2 & 2 & - & 2 & - & u_{l,2} \\
\hline
2 & 12 & - & 12 & u_{r,1} & u_{l,2} \\
\hline
12 & 0 & - & 2 & - & u_{l,2} \\
\hline
12 & 1 & - & 12 & u_{r,1} & u_{l,2} \\
\hline
12 & 2 & u_{l,1} & 2 & - & u_{l,2} \\
\cline{3-6}
 & & u_{l,\star} & 0 & - & - \\
\cline{3-6}
 & & \cellcolor{blau} u_{r,2} & \cellcolor{blau} 2 & \cellcolor{blau} - & \cellcolor{blau} u_{l,2} \\
\hline
12 & 12 & u_{l,1} & 12 & u_{r,1} & u_{l,2} \\
\hline
 & & u_{l,\star} & 1 & u_{r,1} & - \\
\hline
 & & \cellcolor{blau} u_{r,2} & \cellcolor{blau} 12 & \cellcolor{blau} u_{r,1} & \cellcolor{blau} u_{l,2} \\
\hline
 & & u_{r,\star} & 2 & - & u_{l,2} \\
\hline
\multicolumn{6}{c}{ } \\[-6pt]
\multicolumn{6}{c}{\text{Label }\rjoin} 
\end{array}
$ \hfill \null
\caption{The tables for the case that $\hat{l}$ and $\hat{r}$ are inner nodes of $T$.}
\label{table1}
\end{table}

\begin{table}[hbt]
\setlength{\arraycolsep}{4pt}
\null \hfill
$
\begin{array}[t]{|cc|c|ccc|}
\hline
t_{\hat{l}} & t_{\hat{r}} & v & t_{\hat{w}} & u_1 & u_2 \\
\hline
0 & - & - & 2 & - & u_r \\
\hline
1 & - & - & 12  & u_{l,1} & u_r^{\star} \\
\hline
2 & - &  u_{l,2} &2 & - & u_{r} \\
\cline{3-6}
 & & \cellcolor{blau} u_r & \cellcolor{blau} 2 & \cellcolor{blau} - & \cellcolor{blau} u_{l,2} \\
\hline
12 & - & \cellcolor{blau} u_{l,2} & \cellcolor{blau} 12 & \cellcolor{blau} u_{l,1} & \cellcolor{blau} u_r^{\star} \\
\cline{3-6}
 & & \cellcolor{blau} u_{l,\star} & \cellcolor{blau} 2 & \cellcolor{blau} - & \cellcolor{blau} u_r \\
\cline{3-6}
 & & u_r & 2 & - & u_{l,2} \\
\hline
\multicolumn{6}{c}{ } \\[-6pt]
\multicolumn{6}{c}{\text{Label }\union \text{, } \hat{r} \text{ is a leaf}} \\
\end{array}
$
\hfill
$
\begin{array}[t]{|cc|c|ccc|}
\hline
t_{\hat{l}} & t_{\hat{r}} & v & t_{\hat{w}} & u_1 & u_2 \\
\hline
0 & - & - & 1 & u_r & - \\
\hline
1 & - & u_{l,1} & 1 & u_{r} & - \\
\cline{3-6}
 & & \cellcolor{blau} u_{r} & \cellcolor{blau} 1 & \cellcolor{blau} u_{l,1} & \cellcolor{blau} - \\
\hline
2 & - & - & 12 & u_r^{\star} & u_{l,2} \\
\hline
12 & - & \cellcolor{blau} u_{l,1}  & \cellcolor{blau} 12 & \cellcolor{blau} u_r^{\star} & \cellcolor{blau} u_{l,2} \\
\cline{3-6}
 & & \cellcolor{blau} u_{l,\star} & \cellcolor{blau} 1 & \cellcolor{blau} u_r & \cellcolor{blau} - \\
\cline{3-6}
 & & u_r & 1 & u_{l,1} & - \\
\hline
\multicolumn{6}{c}{ } \\[-6pt]
\multicolumn{6}{c}{\text{Label }\join \text{, } \hat{r} \text{ is a leaf}} \\
\end{array}
$ \hfill \null

\setlength{\arraycolsep}{4pt}
\null \hfill
$
\begin{array}[t]{|cc|c|ccc|}
\hline
t_{\hat{l}} & t_{\hat{r}} & v & t_{\hat{w}} & u_1 & u_2 \\
\hline
0 & - & - & 1 & u_r & - \\
\hline
1 & - & \cellcolor{blau} u_{l,1} &  \cellcolor{blau} 1 &  \cellcolor{blau} u_r &  \cellcolor{blau} - \\
\cline{3-6}
  & & u_r & 0 & - & - \\
\hline
2 & - & - & 12 & u_r & u_{l,2} \\
\hline
12  & - & \cellcolor{blau} u_{l,1} & \cellcolor{blau} 12 & \cellcolor{blau} u_r & \cellcolor{blau} u_{l,2} \\
\cline{3-6}
 & & u_{l,\star} & 1 & u_r & - \\
\cline{3-6}
 & & u_r & 2 & - & u_{l,2} \\
\hline
\multicolumn{6}{c}{ } \\[-6pt]
\multicolumn{6}{c}{\text{Label }\rjoin \text{, } \hat{r} \text{ is a leaf}} \\
\end{array}
$
\hfill
$
\begin{array}[t]{|cc|c|ccc|}
\hline
t_{\hat{l}} & t_{\hat{r}} & v & t_{\hat{w}} & u_1 & u_2 \\
\hline
- & 0 & - & 2 & - & u_l \\
\hline
- & 1 & - & 12 & u_{r,1} & u_l \\
\hline
- & 2 & u_l & 0 & - & -\\
\cline{3-6}
 & & \cellcolor{blau} u_{r,2} & \cellcolor{blau} 2 & \cellcolor{blau} - & \cellcolor{blau} u_l \\
\hline
- & 12 & u_l & 1 & u_{r,1} & - \\
\cline{3-6}
 & & \cellcolor{blau} u_{r,2} & \cellcolor{blau} 12 & \cellcolor{blau} u_{r,1} &  \cellcolor{blau} u_l \\
\cline{3-6}
 & & u_{r,\star} & 2 & - & u_l \\
\hline
\multicolumn{6}{c}{ } \\[-6pt]
\multicolumn{6}{c}{\text{Label }\rjoin\text{, }\hat{l}\text{ is a leaf}} \\
\end{array}
$
\hfill \null
\caption{The tables for the case that either $\hat{l}$ or $\hat{r}$ is a leaf of $T$.}
\label{table2}
\end{table}

\begin{table}[hbt]
\setlength{\arraycolsep}{4pt}
\null \hfill
$
\begin{array}[t]{|c|ccc|}
\hline
R_{\hat{w},t_{\hat{w}}} & t_{\hat{w}} & u_1 & u_2 \\
\hline
\{u_l\} & 2 & - & u_r \\
\hline
\rowcolor{blau} \{u_r\} & 2 & - & u_l \\
\hline
\multicolumn{4}{c}{ } \\[-6pt]
\multicolumn{4}{c}{\text{Label }\union \text{, }\hat{l} \text{ and } \hat{r} \text{ are leaves}} \\
\end{array}
$
\hfill
$
\begin{array}[t]{|c|ccc|}
\hline
R_{\hat{w},t_{\hat{w}}} & t_{\hat{w}} & u_1 & u_2 \\
\hline
\{u_l\} & 1 & u_r & - \\
\hline
\rowcolor{blau} \{u_r\} & 1 & u_l & - \\
\hline
\multicolumn{4}{c}{ } \\[-6pt]
\multicolumn{4}{c}{\text{Label }\join \text{, }\hat{l} \text{ and } \hat{r} \text{ are leaves}} \\
\end{array}
$
\hfill
$
\begin{array}[t]{|c|ccc|}
\hline
R_{\hat{w},t_{\hat{w}}} & t_{\hat{w}} & u_1 & u_2 \\
\hline
\{u_l\} & 1 & u_r & - \\
\hline
\{u_r\} & 2 & - & u_l \\
\hline
\multicolumn{4}{c}{ } \\[-6pt]
\multicolumn{4}{c}{\text{Label }\rjoin \text{, }\hat{l} \text{ and } \hat{r} \text{ are leaves}} \\
\end{array}
$
\hfill \null
\caption{The tables for the case that $\hat{l}$ and $\hat{r}$ are leaves of $T$.}
\label{table3}
\end{table}

Initially, all resolving sets $R_{\hat{w},t_{\hat{w}}}$, $\hat{w} \in V(T)$, $t_{\hat{w}} \in \{0,1,2,12\}$, are undefined. The rules for updating the sets are summarised in \Cref{table1,table2,table3}.
The rows $(t_{\hat{l}},t_{\hat{r}},v,t_{\hat{w}},u_1,u_2)$ of \Cref{table1,table2} show how the resolving sets $R_{\hat{w},t_{\hat{w}}}$ for $V(G_{\hat{w}})$ in $G$ are updated.
\begin{itemize}
\item
The types $t_{\hat{l}}$, $t_{\hat{r}}$, and $t_{\hat{w}}$ are the types of the resolving sets for $V(G_{\hat{l}})$, $V(G_{\hat{r}})$, and $V(G_{\hat{w}})$ in $G$, respectively. If $t_{\hat{l}} = "-"$ or $t_{\hat{r}} = "-"$, then $\hat{l}$ or $\hat{r}$ is a leaf of $T$, respectively.
\item
Vertex $v$ is the vertex that can be added to create a resolving set for $V(G_{\hat{w}})$ in $G$ from the union of a minimum resolving set for $V(G_{\hat{l}})$ in $G$ and a minimum resolving set for $V(G_{\hat{r}})$ in $G$. If $v = "-"$, then $R_{\hat{l},t_{\hat{l}}} \cup R_{\hat{r},t_{\hat{r}}}$ is already a resolving set for $V(G_{\hat{w}})$ in $G$. In this case, $R_{\hat{w},t_{\hat{w}}}$ is updated by $R_{\hat{l},t_{\hat{l}}} \cup R_{\hat{r},t_{\hat{r}}}$.
Otherwise, $R_{\hat{l},t_{\hat{l}}} \cup R_{\hat{r},t_{\hat{r}}} \cup \{v\}$ is a resolving set for $V(G_{\hat{w}})$ in $G$ and $R_{\hat{w},t_{\hat{w}}}$ is updated by $R_{\hat{l},t_{\hat{l}}} \cup R_{\hat{r},t_{\hat{r}}} \cup \{v\}$.
Vertex $v$ is either the 1-vertex, 2-vertex, or double remover of $V(G_{\hat{l}})$ w.r.t.\ $R_{\hat{l},t_{\hat{l}}}$ denoted by $u_{l,1}$, $u_{l,2}$, and $u_{l,\star}$, respectively, or the 1-vertex, 2-vertex, or double remover of $V(G_{\hat{r}})$ w.r.t.\ $R_{\hat{r},t_{\hat{r}}}$ denoted by $u_{r,1}$, $u_{r,2}$, and $u_{r,\star}$, respectively. If $\hat{l}$ or $\hat{r}$ is a leaf, then $u_l=\vertex{\hat{l}}$ and $u_r=\vertex{\hat{r}}$, respectively.
\item
Vertex $u_1$ and $u_2$ are the 1-vertex and 2-vertex of $G_{\hat{w}}$ w.r.t.\ the resulting resolving set for $V(G_{\hat{w}})$ in $G$. A created double remover gets a superscript asterisk. If $u_1 = "-"$ or $u_2 = "-"$, then $G_{\hat{w}}$ has no 1-vertex or 2-vertex, respectively, w.r.t.\ the resulting resolving set for $V(G_{\hat{w}})$ in $G$.
\end{itemize}
\Cref{table3} describes the minimum resolving sets for $V(G_{\hat{w}})$ in $G$ for the case that $\hat{l}$ and $\hat{r}$ are leaves in $T$.

The rows with a blue background do not need to be taken into account, because in these cases there is another selection that provides a solution that is at least as good as the blue one.


The running time of computing a minimum resolving set for $G$ is therefore linear in the size of the given co-graph $G$.

%
%

\begin{theorem}
A minimum resolving set for a strongly connected directed co-graph $G$ is computable in linear time.
\label{theorem2}
\end{theorem}

\begin{example}
The table below shows the results computed by the {\bf BottomUp}() procedure for the co-graph $G$ defined by the co-tree $T$ below. There is a column for each possible minimum resolving set $R_{\hat{w},t_{\hat{w}}}$, $t_{\hat{w}} \in \{0,1,2,12\}$. The vertices $u_1,u_2$ next to the minimum resolving sets $R_{\hat{w},t_{\hat{w}}}$ are the 1-vertex and 2-vertices of $G_{\hat{w}}$ w.r.t.\ $R_{\hat{w},t_{\hat{w}}}$. The double remover is marked by an asterisk.

\medskip
\centerline{\includegraphics[width=190pt]{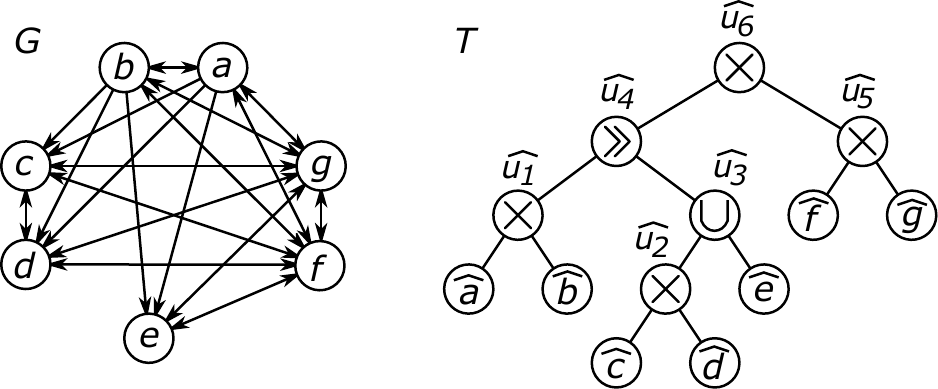}}
$$
\begin{array}{l|llll}
G_{\hat{w}} & R_{\hat{w},0} & R_{\hat{w},1}/u_1 & R_{\hat{w},2}/u_2 & R_{\hat{w},12}/u_1,u_2 \\
\hline

G_{\hat{u}_1}=G_{\hat{\capfont{a}}} \,\join G_{\hat{\capfont{b}}} &
- &
\{\capfont{a}\}/\capfont{b} &
- &
- \\

G_{\hat{u}_2}=G_{\hat{\capfont{c}}} \,\join G_{\hat{\capfont{d}}} &
- &
\{\capfont{c}\}/\capfont{d} &
- &
- \\

G_{\hat{u}_3}=G_{\hat{u}_2} \union \hat{\capfont{e}} &
- &
- &
- &
\{\capfont{c}\}/\capfont{d},\capfont{e}^{\star} \\

G_{\hat{u}_4}=G_{\hat{u}_1} \,\rjoin G_{\hat{u}_3} &
\{\capfont{a},\capfont{c},\capfont{e}\} &
- &
- &
- \\

G_{\hat{u}_5}=G_{\hat{\capfont{f}}} ~\join G_{\hat{\capfont{g}}} &
- &
- &
\{\capfont{f}\}/\capfont{g} &
- \\

G_{\hat{u}_6}=G_{\hat{u}_4} \join G_{\hat{u}_5} &
\{\capfont{a},\capfont{c},\capfont{e},\capfont{f}\} &
- &
- & 
- \\
\end{array}
$$













Vertex set $\{\capfont{a},\capfont{c},\capfont{e},\capfont{f}\}$ is a minimum resolving set for the strongly connected directed co-graph $G$ defined by co-tree $T$ of the figure above.
\label{example1}
\end{example}

\section{Directed acyclic graphs}

An undirected graph can easily be transformed into a directed graph by replacing each undirected edge $\{u,v\}$ by two directed edges $(u,v)$ and $(v,u)$. Thus, {\sc Directed Metric Dimension} is NP-complete, because {\sc Metric Dimension} is NP-complete, see \cite{KRR96}, but it is also NP-complete for oriented graphs, see \cite{RRCM14}. However, the following theorem shows that {\sc Directed Metric Dimension} is also NP-complete for directed acyclic graphs, i.e., for DAGs.

\begin{theorem}
{\sc Directed Metric Dimension} is NP-complete for DAGs.
\label{theorem2}
\end{theorem}

\begin{proof}
The problem is obviously in NP. The NP-hardness is shown by a polynomial time reduction from {\sc Hitting Set}. Let $\mathcal{C}=\{C_1,\ldots,C_m\}$ be a set of subsets of a set $X=\{x_1,\ldots,x_n\}$. We define a directed acyclic graph $G$ such that $G$ has a resolving set of size at most $3+k$ if and only if there is a subset $X' \subset X$ of size at most $k$ such that $X' \cap C_j \not= \emptyset$ for $j=1,\ldots,m$. W.l.o.g., we assume that $m > n$, otherwise duplicate some sets of $\mathcal{C}$, and that for each subset $C_j$ there is at least one $x_i$ such that $x_i \not\in C_j$.

The graph $G$ defined for $X,\mathcal{C}$ has
\begin{enumerate}
\item three vertices $u_a$, $u_b$ and $u_c$,
\item $n$ vertices $u_{x_i}$ for $i=1,\ldots,n$,
\item $2m$ vertices $u_{C_j},u_{C'_j}$ for $j=1,\ldots,m$,
\end{enumerate}
and
\begin{enumerate}
\item $n$ edges $(u_a,u_{x_i})$ for $i=1,\ldots,n$,
\item one edge $(u_b,u_{x_1})$,
\item $n-1$ edges $(u_{x_i},u_{x_{i+1}})$ for $i=1,\ldots,n-1$,
\item two edges $(u_c,u_{C_1}),(u_c,u_{C'_1})$,
\item $m$ edges $(u_{C_j},u_{C'_j})$, for $j=1,\ldots,m$,
\item $4(m-1)$ edges $(u_{C_j},u_{C_{j+1}}),(u_{C_j},u_{C'_{j+1}}),(u_{C'_j},u_{C_{j+1}}),(u_{C'_j},u_{C'_{j+1}})$, for $j=1,\ldots,m-1$,
\item $nm$ edges $(u_{x_i},u_{C_j})$ for $i=1,\ldots,n$ for $j=1,\ldots,m$, and
\item an edge $(u_{x_i},u_{C'_j})$, $1 \leq i \leq n$, $1 \leq j \leq m$, if and only if $x_i \not\in C_j$. 
\end{enumerate}
\Cref{figure7} shows an example.


\begin{figure}[htb]
\center
\includegraphics[width=264pt]{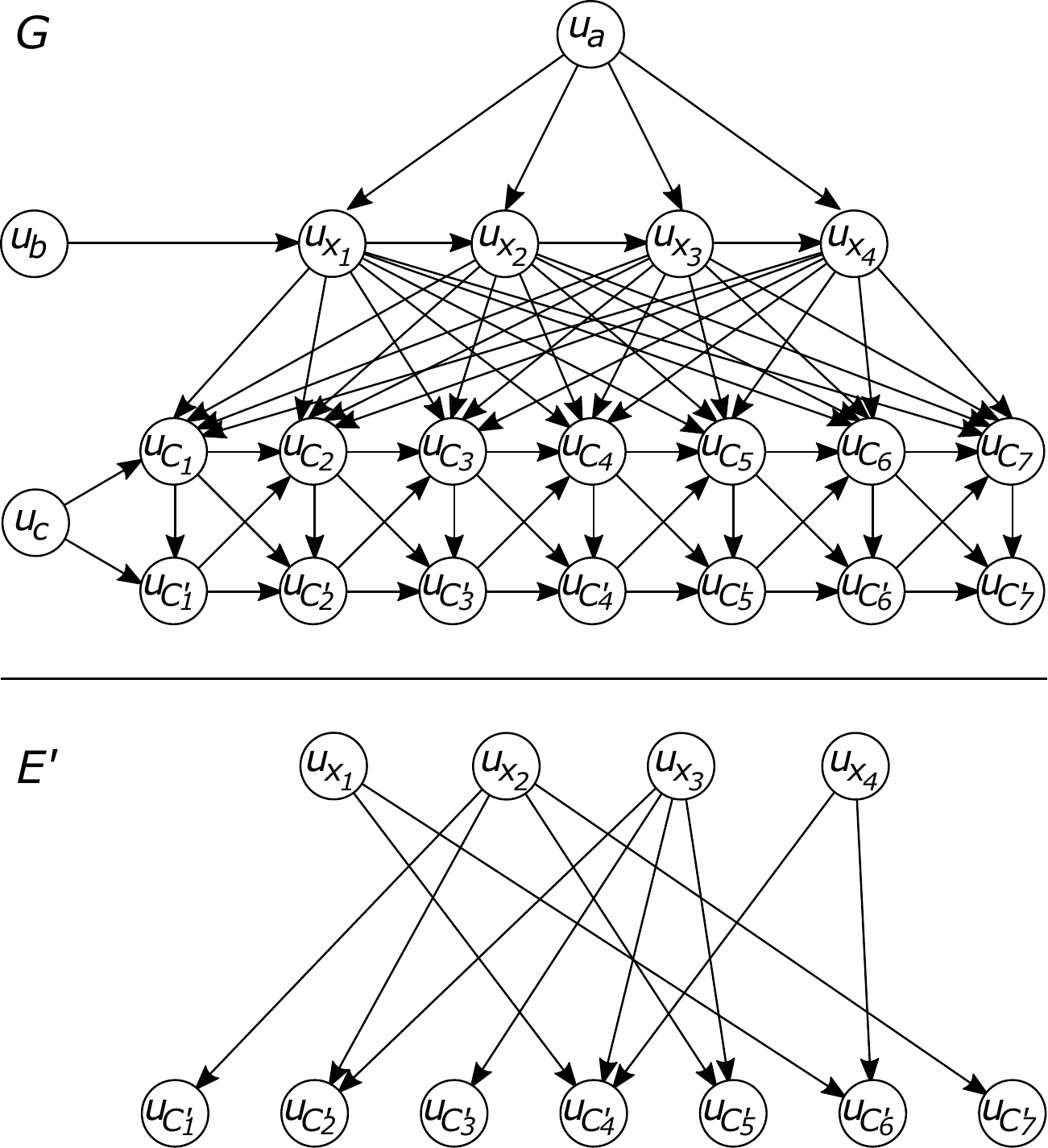} 
\caption{Graph $G$ obtained by instance $X,\mathcal{C}$ with $X=\{x_1,x_2,x_3,x_4\}$, $\mathcal{C}=\{C_1, C_2, C_3, C_4, C_5, C_6, C_7\}$, $C_1=\{x_1,x_3,x_4\}$, $C_2=\{x_1,x_4\}$, $C_3=\{x_1,x_2,x_4\}$, $C_4=\{x_2\}$, $C_5=\{x_1,x_4\}$, $C_6=\{x_2,x_3\}$, $C_7=\{x_1,x_3,x_4\}$ for {\sc Hitting Set} in the proof of \Cref{theorem2}. The edges which result from the membership of the elements $x_i$ to the subset $C_j$ are drawn separately below for reasons of clarity. These are the edges in edge set $E'$. Set $X'=\{x_2,x_4\}$ is a minimum {\sc Hitting set} for $X,\mathcal{C}$, where $\{u_a,u_b,u_c,u_{x_2},u_{x_4}\}$ is a minimum resolving set for $G$.}
\label{figure7}
\end{figure}

Each minimum resolving set for $G$ contains the three vertices $u_a,u_b,u_c$, because these vertices have no incoming edges. Vertex $u_b$ resolves every vertex pair $u_{x_i},u_{x_j}$, $1 \leq i < j \leq n$, because $d_G(u_b,u_{x_i}) = i$. Vertex $u_c$ resolves every vertex pair $(u_{C_i},u_{C_j})$, $(u_{C_i},u_{C'_j})$, $(u_{C'_i},u_{C_j})$ and $(u_{C'_i},u_{C'_j})$, $1 \leq i < j \leq m$, because $d_G(u_c,u_{C_i}) = d_G(u_c,u_{C'_i}) = i$. Vertex $u_a$ resolves every vertex pair $u_{x_i}, u_{C_j}$ and every vertex pair $u_{x_i}, u_{C'_j}$, $1 \leq i \leq n$, $1 \leq j \leq m$, because $d_G(u_a,u_{x_i}) = 1$ and $d_G(u_a,u_{C_j})=2$. Note that we assume that for each subset $C_j$ there is at least one $x_i$ such that $x_i \not\in C_j$, therefore $d_G(u_a,u_{C'_j})=2$. The only vertex pairs that have not been considered yet are the vertex pairs $u_{C_j},u_{C'_j}$ for $j=1,\ldots,m$. A vertex pair $u_{C_j},u_{C'_j}$ can only be resolved by vertex $u_{C_j}$, by vertex $u_{C'_j}$ or by a vertex $u_{x_i}$ such that $x_i \in C_j$, because $x_i \in C_j$ if and only if $d_G(u_{x_i},u_{C'_j})=2$.

Let $X' \subseteq X$ be a hitting set for $\mathcal{C}$ of size at most $k$, i.e., $X' \cap C_j \not= \emptyset$ for $j=1,\ldots,m$. Then the vertices $u_{x_i}$ with $x_i \in X'$ resolve all vertex pairs $u_{C_j}$, $u_{C'_j}$, $1 \leq j \leq m$. Thus $$R=\{u_a,u_b,u_c\} \cup \bigcup_{x_i \in X'} u_{x_i}$$ is a resolving set for $G$ of size at most $3+k$. Let $R$ be a resolving set for $G$ of size at most $3+k$. If $R$ contains a vertex $u_{C_j} \in R$ or $u_{C'_j} \in R$, then replace it by a vertex $u_{x_i}$ where $x_i \in C_j$. The resulting set $R'$ is also a resolving set for $G$ of size at most $3+k$ and $$X'=\{ x_i | u_{x_i} \in R'\}$$ is a hitting set for $\mathcal{C}$. 
\end{proof}


\section{Conclusions}

In this paper we have shown that {\sc Directed Metric Dimension} is decidable in linear time for directed co-graphs and we also presented an algorithm to compute minimum resolving sets for directed co-graphs in linear time. Additionally, we have shown that {\sc Directed Metric Dimension} is NP-complete for DAGs, extending the existing results for general directed and oriented graphs.

The metric dimension as well as its variants have rarely been studied for directed graphs \cite{SW21}. Developing efficient algorithms to compute the metric dimension for specific graph classes is one of the most interesting challenges to us.

\newcommand{\etalchar}[1]{$^{#1}$}


\end{document}